\documentclass[onwcolumn,showpacs]{revtex4-2}
\usepackage{xcolor}
\usepackage{graphicx}
\usepackage{dcolumn}
\usepackage{amsmath}
\usepackage{amssymb}
\usepackage{hyperref}
\usepackage{accents}
\usepackage{mathrsfs}

\usepackage{scalerel}
\usepackage{undertilde}

\allowdisplaybreaks

\newtheorem{theorem}{Theorem}[section]
\newtheorem{corollary}[theorem]{Corollary}
\newtheorem{proposition}[theorem]{Proposition}

\newenvironment{proof}[1][Proof]{\begin{trivlist}
\item[\hskip \labelsep {\bfseries #1}]}{\end{trivlist}}

\newcommand{\qed}{\nobreak \ifvmode \relax \else
	\ifdim\lastskip<1.5em \hskip- \lastskip
	\hskip1.5em plus0em minus0.5em \fi \nobreak
	\vrule height0.75em width0.5em depth0.25em\fi}

\begin{document}

\title{Existence of gradient CKV and gradient conformally stationary LRS spacetimes}

\author{Seoktae \surname{Koh}}
\email{kundol.koh@jejunu.ac.kr}
\affiliation{Department of Science Education, Jeju National University, Jeju, 63243, South Korea}
\affiliation{Institute for Gravitation and the Cosmos, Pennsylvania State University, University Park, PA 16802, USA}
\author{Abbas M \surname{Sherif}}
\email{abbasmsherif25@gmail.com}
\affiliation{Department of Science Education, Jeju National University, Jeju, 63243, South Korea}

\author{Gansukh \surname{Tumurtushaa}}
\email{gansukh@jejunu.ac.kr}
\affiliation{Department of Science Education, Jeju National University, Jeju, 63243, South Korea}

\begin{abstract}
In this work, we study the existence of gradient (proper) CKVs in locally rotationally symmetric spacetimes (LRS), those CKVs in the space spanned by the tangent to observers' congruence and the preferred spatial direction, allowing us to provide a (partial) characterization of gradient conformally static (GCSt) LRS solutions. Irrrotational solutions with non-zero spatial twist admit an irrotational timelike gradient conformal Killing vector field and hence are GCSt. In the case that both the vorticity and twist vanish, that is, restricting to the LRS II subclass, we obtain the necessary and sufficient condition for the spacetime to admit a gradient CKV. This is given by a single wave-like PDE, whose solutions are in bijection to the gradient CKVs on the spacetime. We then introduce a characterization of these spacetimes as GCSt using the character of the divergence of the CKV, provided that the metric functions of the spacetimes obey certain inequalities.
\end{abstract}

\maketitle

\section{Introduction}


The extensive role of symmetries in general relativity is evidenced by their diverse applications, from generating new exact solutions to the Einstein field equations, their use in geometry (conformal geometry), to roles in understanding the thermodynamics of both static and dynamical black holes. The ubiquity of symmetries in general relativity means that new approaches to studying them are being introduced in various contexts. 

Special attention has been devoted to the spherically symmetric solutions to the Einstein field equations. The Bianchi models, for example, have been extensively dealt with in the literature \cite{tsa1,tsa2,kha1,tsa3}, where their conformal algebras have been completely classified. Conformal symmetries of imperfect, perfect, and anisotropic fluids have been extensively analyzed as well \cite{tsa4,tsa5,aac1,aac2,aac3,mahar1,mahar2,mahar3}. 

The class of locally rotationally symmetric (LRS) spacetimes, a particular class of important solutions with diverse applications, which admits specialized decomposition along the timelike and spatial congruences, generalizes spherically symmetric solutions. The existence of (proper) CKVs in these spacetimes, when the vorticities of the timelike and spatial congruences are both non zero, was relatively recently considered by Singh \textit{et al.} \cite{singh1}, where it was shown that these spacetimes always admit a proper conformal killing vector in the subspace spanned by the tangents to the timelike and spacelike congruences. (By proper here it is meant that the divergence of the CKV is non-constant.) The proof largely relied on a defining relationship between the vorticities of the congruences. Van den Bergh \cite{bergh1}, has also considered the case of rotating and twisting LRS spacetimes, where explicit solutions were provided using the existence of such symmetries. (Dyer \textit{et al.} \cite{dy2} have studied such symmetries in some subclass of McVittie metrics \cite{mv1,mv2,mv3}, where conditions on a particular metric function was obtained that ensures the existence of CKVs in these metrics (also see the recent work by Herrera \textit{et al.} \cite{her1}).) In the case of the subclass with both vorticities vanishing (the LRS II class), such relationship is ill-defined and hence, their analysis was not adaptable to this subclass of LRS spacetimes. 

A spacetime which admits a timelike proper CKV is said to be \textit{conformally stationary} (CS), i.e. it is conformal to some stationary spacetime, and if the conformal observers moving along the conformal orbits experience no vorticity, the spacetime is said to be \textit{conformally static} (CSt). If the CKV is a gradient of some smooth function on the spacetime, referred to as the potential function of the vector field, the spacetime is referred to as a gradient conformally stationary (GCS) spacetime (GCSt spacetime, when there is no vorticity associated to the CKV) \cite{rub1}. If the gradient condition is global, then, these spacetimes are stably causal, i.e. the spacetime admits a global time function, and thereby ensures the non-existence of closed timelike curves (see \cite{rub1} and associated references for more details). GCS spacetimes generalizes the wide class of generalized Robertson-Walker (GRW) spacetimes (see the review \cite{man1} and associated references) that includes many spacetime models which are of wide ranging interests to the cosmology and astrophysics community. For GCS spacetimes, in the region that the CKV is timelike, the spacelike slices of the foliation determined by the distribution of the vector field has constant mean curvature (CMC), and is also of immense geometric interest, not least that the CMC condition is useful in the study of the Einstein field equations as an evolution problem (again, see \cite{rub1} and references therein).

In this work, we are interested in characterizing GCS LRS spacetimes. The existence of gradient CKVs have been considered for perfect and anisotropic fluid spacetimes by Daftardar and Dadhich \cite{nd1}, where the local forms of the spacetime metrics admitting the vector field were obtained. The role of gradient CKVs in generating conformal Killing tensors has also been studied extensively examined. For example, using the Koutras algorithm for generating Killing an conformal Killing tensors from existing Killing and conformal Killing vectors on the spacetime \cite{koutras1}, Amery and Maharaj, \cite{Amery1}, provided an explicit construction for the form of Killing tensors in an Einstein spacetime, and additionally demonstrating that the form of the Killing tensors is invariant in the conformally related spacetime. Rani \textit{et al.}, \cite{rani1}, established a generalization of Koutras results by eliminating the orthogonality constraint imposed on the pair of (C)KVs from which the tensor is constructed. Detailed analyses were then carried out on the construction of (conformal) Killing tensors (which the authors referred to as  \textit{(conformal) Killing tensors of gradient type}) from gradient (C)KVs. Our aim here, however, is to provide an analysis of, and establish the necessary and sufficient conditions for the existence of a gradient CKV in the class of LRS spacetimes, and under which conditions are the CKV timelike. 

The approach to be employed here is the 1+1+2 formalism \cite{cc1,cc2}, to be introduced shortly, which specifies the spacetimes in terms of well defined covariant variables. Recently, some of the physical and geometric implications arising from employing this covariant approach have been uncovered. For example, in \cite{chevarra1}, Chevarra \textit{et al.} conducted an extensive analysis of the effects of the 1+1+2 decomposition on spacetimes with conformal symmetry. Various constraints on the spacetime variables were obtained. Fixing an equation of state (assuming a functional dependence of the pressure and energy density) and assuming a perfect fluid matter type, the authors showed that the divergence of CKV obeys a damped wave equation. In \cite{hakata1}, Hakata \textit{et al.} showed that only for a restricted class of equations of state, obeying a non-linear fourth order differential equation, can a shear-free perfect fluid spacetime be homogeneous, thereby justifying the physicality of these solution types. Crucially, the acceleration of the 4-velocity and its expansion are sufficient to fully classify these spacetimes. Chevarra \textit{et al.} recently considered conformal symmetries in the generalized Vaidy metric, demonstrating the dependence of the spacetime variables specified by the temporal and spatial congruences, on the components of the CKV and its divergence \cite{chevarra2}. 

We organize this paper in the following manner. In Section \ref{lrs1}, we introduce the class of LRS spacetimes and briefly discuss the 1+1+2 spacetime decomposition formalism, following the standard literature. In Section \ref{lrs2}, we investigate the existence of gradient CKVs in LRS spacetimes. We decompose our analysis into the LRS II class and non-LRS II class. The GCS non-LRS II case will be completely characterized. In the LRS II case, the necessary and sufficient condition for the spacetime to admit a gradient CKV will be given, and the conditions on the kinematic variables ensuring the timelike character of the CKV is provided, using the character of the divergence of the CKV. We conclude in Section \ref{lrs4} with a summary of our results.


\section{LRS spacetimes and a semi-tetrad decomposition of spacetime}\label{lrs1}

We begin by introducing the \(1+1+2\) semitetrad covariant formalism and locally rotationally symmetric spacetimes in context of the formulation.

\subsection{The 1+1+2 spacetime decomposition}

Like the powerful 1+3 formalism \cite{el3,el4,el5} (also see references therein) which allows for the threading of the spacetime along the fluid flow lines with unit tangent vector \(u^{\mu}\), the 1+1+2 formalism is a special case of the 1+3 formalism, where the 3-space orthogonal to \(u^{\mu}\) is further decomposed along a spacelike unit vector field, which we denote by \(e^{\mu}\), that is orthogonal to \(u^{\mu}\) \cite{cc1,cc2}. Such decomposition allows one to project tensor and vector quantities, as well as the covariant derivative, along the unit directions and the resulting 2 space: the \(\dot{\ }\) notation denotes derivative along the \(u^{\mu}\) direction, \(\hat{\ }\) denotes derivative along the \(e^{\mu}\) direction, and \(\delta_{\mu}=N^{\ \nu}_{\mu}\nabla_{\nu}\) (\(\nabla_{\mu}\) is the 4-dimensional spacetime covariant derivative) the derivative on the 2-space which results from decomposing the 3-space, with \(N^{\mu\nu}\) projecting vectors and tensors orthogonal to \(u^{\mu}\) and \(e^{\mu}\), to the 2-space. (In the literature this 2-space is usually referred to as the ``sheet". This is due to the non-symmetricity, in general, of \(\delta_{\mu}\delta_{\nu}\psi\) for an arbitrary scalar, i.e. \(\delta_{[\mu}\delta_{\nu]}\psi\neq0\), and hence sometimes one has the collection of tangent planes rather than a genuine surface. However, in the case of LRS spacetime, the sheets are always genuine surfaces.) The spacetime metric also decomposes as

\begin{eqnarray*}
g_{\mu\nu}=N_{\mu\nu}-u_{\mu}u_{\nu}+e_{\mu}e_{\nu}.
\end{eqnarray*}

Spacetime vectors can accordingly be decomposed along the two preferred directions and the sheet. For example, a vector \(\psi^{\mu}\) can be written in the form

\begin{eqnarray*}
\psi^{\mu}=\psi_1u^{\mu}+\psi_2e^{\mu}+\bar{\psi}^{\mu},
\end{eqnarray*}
with scalars \(\psi_1\) and \(\psi_2\) being the respective components along \(u^{\mu}\) and \(e^{\mu}\), and \(\bar{\psi}^{\mu}\) denoting the part of \(\psi^{\mu}\) lying in the sheet. 

The Bianchi and Ricci identities split along the preferred directions so that the field equations can be written as a collection of evolution and propagation equations of covariant quantities, along with a set of constraints obtained via projection with \(N^{\mu\nu}\).

The field equations in terms of these covariant quantities, restricted to the LRS case (of interest to this work), are given in Appendix A. In addition, the following relation is useful in the consistency check of the fields equations: Given any scalar \(\psi\) in the spacetime, the dot and hat derivatives commute as

\begin{eqnarray}\label{jun42}
\hat{\dot{\psi}}-\dot{\hat{\psi}}=-A\dot{\psi}+\left(\frac{1}{3}\Theta+\Sigma\right)\hat{\psi}.
\end{eqnarray}
Similar commutation relations for the dot and delta derivatives, the hat and delta derivatives, etc., acting on scalars as well as vectors and tensors, have been obtained (see the references \cite{cc1,cc2}). These are not needed for the purpose of this work and therefore will not be included here.

\subsection{LRS spacetimes}

\textit{Locally rotationally symmetric (LRS)} spacetimes are those admitting a multiply transitive isometry group, with a continuous isotropy group at each point of the spacetime. Locally, these spacetimes admit a preferred spatial direction (see the references \cite{el1,el2} for more details). Due to the symmetry of these spacetimes, one may write the local metric in coordinates \((t,\mathcal{R},y,z)\) as (see \cite{el2} for example)

\begin{align}\label{fork}
ds^2=-A^2dt^2+ B^2dr^2+ C^2dy^2+\left(\left(DC\right)^2+\left(Bh\right)^2-\left(Ag\right)^2\right)dz^2+2\left(A^2gdt-B^2hdr\right)dz,\nonumber
\end{align}
where \(A,B\) and \(C\) are functions of coordinates \(t\) and \(r\), and the functions \(g\) and \(h\) are functions of \(y\). The function \(D\) is a function of \(y\) and \(k\) where \(k\) is a constant that fixes the function \(D\) (\(k=-1\) corresponds to \(\sinh y\), \(k=0\) corresponds to \(y\), \(k=1\) corresponds to \(\sin y\)). In the limiting case that \(g=0=h\), we recover the well
studied LRS II class of spacetimes, which generalizes spherically symmetric solutions to the Einstein field equations. The time and spacelike congruences are specified by the unit directions

\begin{eqnarray*}
u^a=-A^{-1}\partial_t^a;\qquad e^a=B^{-1}\partial_r^a.
\end{eqnarray*}

These spacetimes can be specified entirely by the set of covariant scalars \cite{cc2}

\begin{eqnarray*}
\mathcal{D}:\equiv\lbrace{\rho,p,Q,\Pi,\mathcal{E},\mathcal{H},\mathcal{A},\Theta,\Sigma,\phi,\Omega,\xi\rbrace},
\end{eqnarray*}
where \(\rho\equiv T_{\mu\nu}u^{\mu}u^{\nu}\) is the energy density, \(p\equiv\left(1/3\right)h^{\mu\nu}T_{\mu\nu}\) is the pressure (isotropic), \(Q=-T_{\mu\nu}e^{\mu}u^{\nu}\) is the heat flux, \(\Pi=T_{\mu\nu}e^{\mu}e^{\nu}-p\) is the anisotropic stress, \(\mathcal{E}=E_{\mu\nu}e^{\mu}e^{\nu}\) encodes the electric part of the Weyl tensor \(E_{\mu\nu}\), \(\mathcal{H}=H_{\mu\nu}e^{\mu}e^{\nu}\) (encodes the magnetic part of the Weyl tensor), \(\mathcal{A}=\dot{u}_{\mu}e^{\mu}\) is the acceleration, \(\Theta\equiv D_{\mu}u^{\mu}\) is the expansion, \(\Sigma=e^{\mu}e^{\nu}D_{\langle \nu}u_{{\mu}\rangle}\) is the shear, \(\phi=\delta_{\mu}e^{\mu}\) denotes the expansion of the \(2\)-space (referred to as the sheet expansion), \(\Omega=\omega^{\mu}e_{\mu}\) (\(\omega_{\mu}\) is the vorticity, and \(\xi=(1/2)\varepsilon^{\mu\nu}\delta_{\mu}e_{\nu}\) is the spatial twist (the twist of \(e^{\mu}\)), with

\begin{align}
T_{\mu\nu}=\rho u_{\mu}u_{\nu}+\left(p+\Pi\right)e_{\mu}e_{\nu}+ 2Qe_{(\mu}u_{\nu)} + \left(p-\frac{1}{2}\Pi\right)N_{\mu\nu},
\end{align}
being the stress energy tensor, and the derivative operator \(D_{\mu}\) denoting the covariant derivative on the hypersurface to which \(u^{\mu}\) is orthogonal. 

The covariant derivatives of the unit vector fields \(u^{\mu}\) and \(e^{\mu}\) for LRS spacetimes are given by

\begin{align}
\nabla_{\mu}u_{\nu}&=-\mathcal{A}u_{\mu}e_{\nu}+\left(\frac{1}{3}\Theta+\Sigma\right)e_{\mu}e_{\nu}+\frac{1}{2}\left(\frac{2}{3}\Theta-\Sigma\right)N_{\mu\nu}+\Omega\varepsilon_{\mu\nu},\label{start10}\\
\nabla_{\mu}e_{\nu}&=-\mathcal{A}u_{\mu}u_{\nu}+\left(\frac{1}{3}\Theta+\Sigma\right)e_{\mu}u_{\nu}+\frac{1}{2}\phi N_{\mu\nu}+\xi\varepsilon_{\mu\nu}.\label{start11}
\end{align}

The tensor \(\varepsilon_{\mu\nu}\) is the two-dimensional alternating tensor defined as

\begin{eqnarray*}
\varepsilon_{\mu\nu}=\varepsilon_{\mu\nu\sigma}e^{\sigma}=u^{\delta}\eta_{\delta\mu\nu\sigma}e^{\sigma},
\end{eqnarray*}
with  \(\varepsilon_{\mu\nu\delta}\) and \(\eta_{\delta\mu\nu\sigma}\) respectively denoting the 3-dimensional and 4-dimensional alternating tensors, and where 

\begin{eqnarray*}
\varepsilon_{\mu\nu\delta}=e_{\mu}\varepsilon_{\nu\delta}-e_{\nu}\varepsilon_{\mu\delta}+e_{\delta}\varepsilon_{\mu\nu}.
\end{eqnarray*}

For the rest of this work, by scalar we will always mean those in the spacetime respecting the LRS symmetries (or some combination thereof). That is, those in the covariant set \(\mathcal{D}\).


\section{Conformal symmetries in LRS spacetimes}\label{lrs2}


In this section we consider the existence of gradient CKVs in LRS spacetimes. We will introduce conformal symmetries in LRS spacetimes in context of the 1+1+2 decomposition. 

A spacetime \((\mathcal{M},g_{\mu\nu})\) is said to admit a conformal symmetry if there exists a vector field \(\eta^{\mu}\) such that the flow of the metric along \(\eta^{\mu}\) scales the metric, i.e.

\begin{eqnarray}\label{vec01}
\mathcal{L}_{\eta}g_{\mu\nu}=2\varphi g_{\mu\nu},
\end{eqnarray}
for some smooth function \(\varphi\) on \(\mathcal{M}\), where \(\mathcal{L}_{\eta}\) denotes the Lie derivative along \(\eta\). The vector \(\eta^{\mu}\) is called a conformal Killing vector (CKV) and \(\varphi\) is the associated conformal factor. In the particular case that \(\varphi\) is zero or a non-zero constant, \(\eta^{\mu}\) is a Killing vector (KV) or a homothetic Killing vector (HKV). Otherwise, \(\eta^{\mu}\) is referred to as a \textit{proper} CKV. 

The equation \eqref{vec01} can be expressed as

\begin{eqnarray}\label{vec0}
\nabla_{(\mu}\eta_{\nu)}=\varphi g_{\mu\nu},
\end{eqnarray}
called the conformal Killing equation (or CKE for short), with the round brackets indicating symmetrization on the indices.

Now, the covariant derivative of any smooth vector field \(\mathcal{Z}^{\mu}\) admits the decomposition

\begin{eqnarray}\label{to1}
\nabla_{\nu}\mathcal{Z}_{\mu}=\frac{1}{2}\mathcal{L}_{\mathcal{Z}}g_{\mu\nu}+\mathcal{F}_{\mu\nu},
\end{eqnarray}
where \(\mathcal{F}_{\mu\nu}=-\mathcal{F}_{\nu\mu}\). In the literature the tensor \(\mathcal{F}_{\mu\nu}\) is sometimes referred to as the \textit{conformal bivector} associated to \(\mathcal{Z}_{\mu}\). Indeed, it follows that if \eqref{vec01} holds, then, substituting in \(\eta_{\mu}\) and using \eqref{vec0} the equation \eqref{to1} becomes 

\begin{eqnarray}\label{to2}
\nabla_{\nu}\eta_{\mu}=\varphi g_{\mu\nu}+\mathcal{F}_{\mu\nu},
\end{eqnarray}
which is an equivalent definition of a CKV. If \(\varphi=0\), then \(\nabla_{\mu}\eta_{\nu}\) is antisymmetric. Thus, finding a vector field whose covariant derivative is an antisymmetric tensor field may be equated to finding a KV on the spacetime. (See \cite{hall1} for details and convention of indexing):

We will seek a CKV in the subspace spanned by the fluid flow velocity and the preferred spatial direction:

\begin{eqnarray}\label{vec1}
x^{\mu}=\alpha_1u^{\mu}+\alpha_2e^{\mu},
\end{eqnarray}
for some smooth functions \(\alpha_i\) on the spacetime. In this case, with the help of \eqref{start10} and \eqref{start11}, expanding the CKE \eqref{vec0} and then contracting with \(u^{\mu}u^{\nu},e^{\mu}e^{\nu},u^{(\mu}e^{\nu)}\) and \(N^{\mu\nu}\) results in the covariant set of PDEs in \(\alpha_i\):

\begin{align}
\varphi&=\dot{\alpha}_1+\mathcal{A}\alpha_2,\label{vec3}\\
\varphi&=\hat{\alpha}_2+\left(\frac{1}{3}\Theta+\Sigma\right)\alpha_1,\label{vec4}\\
0&=\dot{\alpha}_2-\hat{\alpha}_1+\mathcal{A}\alpha_1-\left(\frac{1}{3}\Theta+\Sigma\right)\alpha_2,\label{vec5}\\
2\varphi&=\alpha_1\left(\frac{2}{3}\Theta-\Sigma\right)+\alpha_2\phi.\label{vec6}
\end{align}

In \cite{singh1}, it is this set of PDEs that were analyzed to check the existence CKV for LRS solutions that are simultaneously rotating and twisting, where it was established that there always exists a CKV of the form \eqref{vec1}. In the next section, we analyze the case of gradient CKV for the entire class of LRS spacetimes.

\section{Main results}\label{mainr1}


We present our main results in this section. We are interested in studying the existence of gradient CKV in LRS spacetimes, with particular interest in the proper conformal Killing case. This will be followed by the characterization of these spacetimes as GCS.

\subsection{Existence}

Let  \(x^{\mu}\) be a gradient CKV. Then, there exists a scalar \(\Psi\) (which we refer to, henceforth, as the potential function) such that \(\nabla_{\mu}\Psi=x_{\mu}\). Thus, one has

\begin{eqnarray}\label{avec6}
\dot{\Psi}=-\alpha_1\quad\mbox{and}\quad\hat{\Psi}=\alpha_2.
\end{eqnarray}
And since \(\nabla_{[\mu}\nabla_{\nu]}\psi=0\) for any scalar \(\psi\), the conformal bi-vector \(\mathcal{F}_{\mu\nu}\) must vanish. It is a straightforward exercise to check that

\begin{eqnarray}\label{haha3}
\mathcal{F}_{\mu\nu}=-2\left(\dot{\alpha}_2+\alpha_1\mathcal{A}\right)u_{[\mu}e_{\nu]}+\left(\alpha_1\Omega+\alpha_2\xi\right)\varepsilon_{\mu\nu},
\end{eqnarray}
thereby giving the pair of constraints as necessary and sufficient for the CKV to be gradient:

\begin{align}
0&=\dot{\alpha}_2+\alpha_1\mathcal{A},\label{new1}\\
0&=\alpha_1\Omega+\alpha_2\xi.\label{new2}
\end{align}

In the case that at  least one of \(\xi\) or \(\Omega=0\) is non-zero, we establish the following 

\begin{theorem}\label{tht}
For the LRS class of spacetimes with at least one of the rotation or twist non-vanishing, only the subclass with vanishing rotation and non-zero twist admits a timelike gradient CKV in the subspace spanned by the canonical unit directions \(u^a\) and \(e^a\).
\end{theorem}
\begin{proof}

We begin with the case with simultaneously \(\Omega\neq0\) and \(\xi\neq0\). For these spacetimes, the ratio of the rotation and spatial twist \(\Omega/\xi\) obeys the particular constraint \cite{singh2}

\begin{eqnarray}\label{rt4}
\frac{\Omega}{\xi}=\frac{\phi}{\left(\frac{2}{3}\Theta-\Sigma\right)},
\end{eqnarray}
which is well defined and non-zero; either both numerator and denominator are simultaneously zero or, both are non-zero. We can disregard the former as this gives \(\varphi=0\) from \eqref{vec6}.

Now, for an arbitrary scalar \(\psi\), taking the covariant twice and contracting with \(\varepsilon_{\mu\nu}\) gives 

\begin{eqnarray}\label{oig0}
\Omega\dot{\psi}=\xi\hat{\psi}.
\end{eqnarray}
This, taken along with \eqref{new2}, implies that the pair \((\alpha_1,\alpha_2)\) have the following possible solutions set

\begin{eqnarray}\label{rt2}
(\alpha_1,\alpha_2)=\lbrace{(\xi,-\Omega),(-\xi,\Omega)\rbrace}.
\end{eqnarray}
Therefore, the possible candidates for GCKV in these LRS solutions should have the form

\begin{eqnarray}\label{rt20}
x^{\mu}=\pm\xi\left(u^{\mu}-\frac{\Omega}{\xi}e^{\mu}\right).
\end{eqnarray}
Using either of the pair in \eqref{rt2} and \eqref{new1} gives

\begin{eqnarray}\label{rt20}
\dot{\Omega}=\mathcal{A}\xi,
\end{eqnarray}
which we compare to \eqref{aevo1} and get

\begin{eqnarray}\label{rt3}
\left(\frac{2}{3}\Theta-\Sigma\right)\Omega=0.
\end{eqnarray}
By assumption, \(\Omega\neq0\), and therefore we must have \((2/3)\Theta-\Sigma=0\), which we have already ruled as not possible.

In the case that exactly one of \(\xi\) or \(\Omega\) is zero, we have only two possible configurations satisfying \(\mathcal{F}_{\mu\nu}=0\): \(\alpha_1=\xi=0;\ \Omega\neq0\) and \(\alpha_2=\Omega=0;\ \xi\neq0\). In the case of the former, the GCKV is spacelike and can be disregarded. 

For the case \(\alpha_2=\Omega=0;\ \xi\neq0\), the vector field \(x^{\mu}\) is clearly timelike. Indeed, \eqref{new2} is obviously satisfied. As these spacetimes are spatially homogeneous, they cannot accelerate as can be seen from \eqref{aevo1} so that \eqref{new1} also holds. (It can also be checked by comparing \eqref{vec4} and \eqref{vec6} that the spacetime has to be shear free. By \eqref{evo2} the spacetime is also necessarily non-dissipative). The component \(\alpha_1\) can be obtained by solving the ODE

\begin{eqnarray}\label{oig1}
\dot{\alpha}=\frac{1}{3}\Theta\alpha,
\end{eqnarray}
with the associated conformal factor to the CKV given by

\begin{eqnarray*}
\varphi=\frac{1}{3}\alpha\Theta.
\end{eqnarray*}
This concludes the proof.\qed
\end{proof}

One computes the vorticity of \(x^a\) as

\begin{align}\label{thereis4}
\bar{\omega}^{\mu}_x&=*\mathcal{F}^{\nu\mu}x_{\nu}\nonumber\\
&=\left(\alpha_1\Omega+\alpha_2\xi\right)\left(\alpha_2u^{\mu}+\alpha_1e^{\mu}\right),
\end{align}
with \(*\mathcal{F}^{\nu\mu}\) denoting the left dual of \(\mathcal{F}^{\nu\mu}\). Clearly, the case \(\Omega=0;\ \xi \neq0\) leads to a vanishing vorticity. (The case of LRS II also follows, and hence, a GCS LRS spacetime is necessarily GCSt.) As a corollary to Theorem \ref{tht}, it indeed follows that

\begin{corollary}\label{cor1}
An irrotational LRS spacetime with a non-vanishing twist is a GCSt spacetime.
\end{corollary}

In the case \(\alpha_1=\xi=0;\ \Omega\neq0\), where the CKV is spacelike, it is seen that the CKV is in fact a KV, i.e. the spacetime is static. This is seen from \eqref{oig0}, where for all scalars \(\psi\), \(\dot{\psi}=0\). (That the spacetime is static may also be found by comparing \eqref{aevo1} and \eqref{aevo5}).

Of course, for the case of LRS II spacetimes the additional relation \eqref{oig0} cannot be utilized, and therefore no immediate statements about existence of CKV of the type considered in this work can be made, as was noticed in \cite{singh1}. However, we will obtain existence results for GCKV in the LRS II case, and provided that the timelike criterion \((\alpha_2^2/\alpha_1^2)<1\) holds, this would characterize GCS LRS II spacetimes.

Suppose that a LRS II spacetime admits a gradient vector field \(x^a\). We will assume that the potential function \(\Psi\) is at least twice differentiable. Firstly, adding \eqref{vec3} and \eqref{vec4}, and then comparing to \eqref{vec6}, we obtain the following

\begin{eqnarray}\label{vec7}
\dot{\alpha}_1+\hat{\alpha}_2=\alpha_1\left(\frac{1}{3}\Theta-2\Sigma\right)-\alpha_2\left(\mathcal{A}-\phi\right).
\end{eqnarray}

Additionally, from \eqref{vec5}, imposing 

\begin{eqnarray}\label{vecz1}
\hat{\alpha}_1+\left(\frac{1}{3}\Theta+\Sigma\right)\alpha_2=0,
\end{eqnarray}
would ensure that \(\mathcal{F}_{ab}=0\). Thus, a vector field \(x^a\) of the form \eqref{vec1}. Then, \(x^a\) is a gradient CKV for \(\mathcal{M}\) provided its components verify \eqref{vec7} and \eqref{vecz1}.

Now, the substitution of \eqref{avec6} into \eqref{vec7} therefore gives the wave-like PDE

\begin{eqnarray}\label{vec07}
-\ddot{\Psi}+\hat{\hat{\Psi}}+\left(\frac{1}{3}\Theta-2\Sigma\right)\dot{\Psi}+\left(\mathcal{A}-\phi\right)\hat{\Psi}=0.
\end{eqnarray}
Furthermore, \eqref{vecz1} is just 

\begin{eqnarray}\label{vec08}
\hat{\dot{\Psi}}-\left(\frac{1}{3}\Theta+\Sigma\right)\hat{\Psi}=0,
\end{eqnarray}
and so, a gradient vector field with potential function \(\Psi\) is a CKV provided \(\Psi\) obeys both \eqref{vec07} and \eqref{vec08}. As its easily seen, \eqref{vec5} is just the commutation relation \eqref{jun42} for \(\Psi\), which always holds true. And so, for a gradient vector field with the \(\Psi\) potential, the condition \eqref{vec07} on \(\Psi\) is both necessary and sufficient for the gradient vector field to be a CKV, thereby allowing us to state the following result:

\begin{theorem}\label{aap2}
Let \(\mathcal{M}\) be a LRS II spacetime, and let \(x^a\) be a gradient vector field of the form \eqref{vec1} in \(\mathcal{M}\), whose potential function \(\Psi\) is at least \(\mathcal{C}^2\)-differentiable. Then, \(x^a\) is a CKV if and only if \(\Psi\) verifies \eqref{vec07}.
\end{theorem}

By standard arguments from the theory of hyperbolic PDEs, we know that \eqref{vec07} admits a unique solution, subject to an initial data specified on a Cauchy hypersurface. Indeed, Cauchy hypersurfaces exist in LRS spacetimes, induced by the canonical splitting, i.e. surfaces of constant \(t\), on which one may specify an initial data to uniquely solve \eqref{vec07}. Thus, \eqref{vec07} admits a unique solution \(\Psi_0\), subject to an initial data. And since \eqref{vec07} is both necessary and sufficient as per Theorem \ref{aap2}, this allows us to state the following:

\begin{theorem}\label{aap3}
Any solution \(\Psi_0\) to the PDE \eqref{vec07} in a LRS II spacetime \(\mathcal{M}\), uniquely determines the gradient CKV \(x^a=\nabla^a\Psi_0\) on \(\mathcal{M}\).
\end{theorem}

While the above theorems allows for considerations for more general LRS II spacetimes, we emphasize that the equations \eqref{vec07} and \eqref{vecz1} equally works for static and spatially homogeneous solutions, in which case the above theorems do not necessarily hold. In any case, let us consider a simple example, the Robertson-Walker solution. Indeed, in either the static or the spatially homogeneous case, \eqref{vecz1} trivially holds. In the spatially homogeneous case, the sheet expansion vanishes and the spacetime is non-accelerating and shear-free. From the metric \eqref{fork}, one computes for Robertson-Walker (\(A=1\) and \(B=C=a(t)\) is the scale factor)

\begin{eqnarray}\label{thet1}
\Theta=\frac{1}{A}\left(\frac{B_t}{B}+2\frac{C_t}{C}\right)=3\frac{a_t}{a},
\end{eqnarray}
where the \(t\)-subscript denotes partial derivative with respect to \(t\). The equation \eqref{vec07} in coordinates is then simply

\begin{eqnarray*}
\Psi_{tt}+\frac{a_t}{a}\Psi_t=0,
\end{eqnarray*}
so that \(\Psi_t=1/a\), thereby giving the gradient CKV

\begin{eqnarray*}
x^a=\frac{1}{a}u^a.
\end{eqnarray*}
We note that our convention is \(u^a=-\partial_t^a\), which is different from the often used convention \(u^a=\partial_t^a\), in which case instead of \(1/a\) one would simply have \(a\) as the vector component. Similarly, the associated conformal factor will be different depending on the convention used. 

As is seen from the previous discussion, we have a one-to-one (bijective) correspondence between the solutions to \eqref{vec07} and gradient CKVs in LRS II spacetimes. That is, a solution to \eqref{vec07} on a constant \(t\) hypersurface in a LRS II spacetime specifies a gradient CKV for that spacetime. We draw a parallel with the well studied \textit{Killing initial data} equations \cite{kid1} and the relatively recently introduced extension, \textit{conformal Killing initial data} equations \cite{kid2} (a related notion of \textit{conformal Killing-Yano initial data} equations has very recently been introduced in the reference \cite{kid3}). These are a set of PDEs on the background of an initial data hypersurface whose solutions are in one-to-one correspondence with KVs (CKVs) of the spacetime evolved from the hypersurface. Rather than a set of PDEs, we have a single PDE, interestingly, with this property. In keeping with the nomenclature, one may consider \eqref{vec07} a LRS \textit{gradient conformal Killing initial data} (GCKID) equation.

\subsection{The timelike criterion: Gradient conformal staticity}

Now that we have established the criterion for a LRS II spacetime to admit a gradient CKV, we turn to the condition(s) under which such spacetimes are GCSt. We will approach this problem by appealing to the character of the factor \(\varphi\).

It can be easily checked that the PDE \eqref{vec07} is just the expression

\begin{eqnarray}\label{poi1}
\Box\Psi=4\varphi,
\end{eqnarray}
where the operator \(\Box=\nabla_{\mu}\nabla^{\mu}\) is the usual d'Alembertian operator. The above form of the GCKID equation allows us to draw some immediate conclusion about the solutions, based on the character of the divergence $\varphi$. Indeed, if one is concerned with the non-Killing case as it is in this work, one can immediately rule out the class of harmonic functions as solutions to GCKID equation. This is to say that, a harmonic function is a solution to the equation \eqref{vec07} if and only if the GCKV is a KV. In other words, \textit{the potential function \(\Psi\) of a gradient KV in LRS II spacetime must satisfy the scalar wave equation \(\Box\Psi=0\)}.

Additional properties of the solution \(\Psi\) can be gleaned from the behavior of conformal observers -- whether they are converging or diverging -- using properties of subharmonic (\textit{resp.} superharmonic) functions (\(\Box\Psi\geq0\) (\textit{resp.} \(\Box\Psi\leq0\))). More particularly, if the conformal observers are diverging, i.e. \(\varphi>0\), we rule out superharmonic functions as solutions to the GCKID equation. And if conformal observers are converging, i.e. \(\varphi<0\), we rule out subharmonic functions as solutions to the GCKID equation. 

Now, the timelike condition on a gradient CKV requires the bound

\begin{eqnarray}
\frac{\dot{\Psi}^2}{\hat{\Psi}^2}>1.
\end{eqnarray}

For a converging conformal observer \(\varphi<0\), we have (from \eqref{vec5}) 

\begin{eqnarray}\label{pot1}
\frac{\dot{\Psi}^2}{\hat{\Psi}^2}>\left(\frac{\phi}{\frac{2}{3}\Theta-\Sigma}\right)^2.
\end{eqnarray}
It therefore follows that the condition

\begin{eqnarray}\label{pott}
\left(\frac{\phi}{\frac{2}{3}\Theta-\Sigma}\right)^2\geq 1,
\end{eqnarray}
suffices in order for the timelike criterion on the GCKV to hold. In terms of the metric functions (taking positive root), this condition reduces to the following restriction on the metric function (it can be computed that \(\phi=2C_r/(BC)\) and \(\Sigma=2(B_t/B-C_t/C)/3A\), for a LRS II spacetime):

\begin{eqnarray}\label{pot2}
C_r\geq C_t\frac{B}{A}.
\end{eqnarray}

On the other hand, for a diverging conformal observer \(\varphi>0\), we have that

\begin{eqnarray}\label{pot}
\frac{\dot{\Psi}^2}{\hat{\Psi}^2}<\left(\frac{\phi}{\frac{2}{3}\Theta-\Sigma}\right)^2,
\end{eqnarray}
and hence, to satisfy the timelike criterion, it is both necessary and sufficient to have

\begin{eqnarray}
\left(\frac{\phi}{\frac{2}{3}\Theta-\Sigma}\right)^2>1,
\end{eqnarray}
which in terms of the metric functions is

\begin{eqnarray}\label{pot3}
C_r^2> C_t^2\frac{B^2}{A^2}.
\end{eqnarray}
Thus, we have the following characterizations of GCSt LRS II solutions:

\begin{proposition}\label{pro1}
Any LRS II spacetime admitting a superharmonic solution to the GCKID equation, with metric functions satisfying \eqref{pot2}, is a GCSt spacetime.
\end{proposition}

\begin{proposition}\label{pro2}
Any LRS II spacetime admitting a subharmonic solution to the GCKID equation, is a GCSt spacetime if and only if its metric functions satisfy \eqref{pot3}.
\end{proposition}

Let us give a simple example. We do not provide a solution for the GCKID equation, but assuming a solution exists, we inspect the timelike criterion. Consider for example, a LTB-type metric

\begin{eqnarray}\label{ltb1}
ds^2=-dt^2+(R_rf)^2dr^2+R^2d\bar{\Omega}^2,
\end{eqnarray}
with \(d\bar{\Omega}^2\) denoting the 2-sphere metric, where \(f=f(r)>0\), and assume \(R=R(t,r)\) splits as \(R=rR_1(t)\), with \(-\dot{R}_1=R_{1t}<0\). (There are LTB solutions evolved from an initial constant \(t\)-slice, having this property and admitting a black hole, see for example \cite{ib3}.) Then, the conditions \eqref{pott} and \eqref{pot3} are the following respective bounds on the radial coordinate:

\begin{subequations}
\begin{align}
r&\leq \frac{1}{R_{1t}f},\label{pot5}\\
r&<-\frac{1}{R_{1t}f}.\label{pot6}
\end{align}
\end{subequations}
Clearly, \eqref{pot5} is not possible since \(r>0\), suggesting that if the metric \eqref{ltb1} admits a solution to the GCKID equation determining a timelike gradient CKV, the solution is necessarily superharmonic, with \(\varphi>0\). Furthermore, there are trapped surfaces behind the horizon \(r=2m(r)\), where \(m(r)\) is the Misner-Sharp mass, thereby imposing the upper bound \(m(r)<-1/(2R_{1t}f)\) on the mass in the timelike region.


\section{Discussion}\label{lrs4}


We have carried out a thorough analysis of the existence of gradient CKV in LRS spacetimes, and the characterization of LRS spacetimes as gradient conformally stationary (static). We considered CKVs that lie in the subspace spanned by the tangents to the timelike and spatial congruences in the spacetime, employing the 1+1+2 semitetrad covariant formalism, particularly adapted to analyzing LRS spacetimes.

The existence of gradient CKV of the form considered here was ruled out for the case where the LRS spacetime is simultaneously rotating and twisting. Hence, while it was demonstrated in \cite{singh1} that these spacetimes always admit CKV in the subspace under consideration, the gradient case is ruled out. It was consequently shown that for LRS spacetimes with at least one of the rotation or twist vanishing, only those twisted solutions with vanishing rotation can admit a timelike gradient CKV. The straightforward analysis of existence of GCKV for the case with at least one of the rotation or twist vanishing was made possible due to the particular relationship \(\Omega\dot{\psi}=\xi\hat{\psi}\) for scalars \(\psi\) in the spacetime. In the LRS II case, such relationship fails and hence not applicable. However, it is noticed that the constraint equation of the conformal Killing equation is the commutation relation for the potential function \(\Psi\) in the gradient case, and hence, the remaining three PDEs of the conformal Killing equations are sufficient. It turns out that the three PDEs can be combined into a single wave-like PDE for \(\Psi\), whose solutions (initial data can be specified on the constant \(t\) Cauchy surfaces in LRS II spacetimes to ensure a solution to the PDE) are in bijection with the gradient CKV on the spacetime, allowing us to draw an analogy to the well studied notions of Killing and conformal Killing initial data equations, and naming the equation a LRS gradient CKID (GCKID) equation. These results in some sense extends the work of \cite{singh1} by also considering existence in the LRS II case.

Given that the conformal Killing equations, using the usual `box' operator, takes the form \(\Box\Psi=4\varphi\), knowing the character of the conformal observers also allows us to immediately rule out certain functions on the spacetime as solutions to the GCKID equation, and introduce a characterization of the spacetimes as GCSt. Of course, harmonic functions as solutions are ruled out. The sign of the conformal factor, which tells us whether the conformal observers are converging or diverging, determines whether a solution to the PDE is a subharmonic or superharmonic function: the solution must be subharmonic for diverging conformal observers and superharmonic for converging conformal observers. For each class of conformal observers, the timelike criterion for the CKV is analyzed, and given by very specific conditions on the metric functions of the spacetime: for diverging conformal observers, the timelike criterion holds if and only if the metric functions of the spacetime obey the inequality $C_r^2>C_t^2(B^2/A^2)$. And in the case of converging conformal observers, a sufficient condition for the timelike criterion to hold is that the metric functions of the spacetime obey the inequality $C_r\geq C_t(B/A)$.

A simple demonstrable example shows that for an LTB-type metric, any `timelike' solution to the GCKID equation must be a superharmonic function on the spacetime, with an upper bound on the Misner-Sharp mass holding in the timelike region.

The factor $\varphi$ is an important ingredient in many physical processes in spacetimes, when a CKV exists. For example, the convergence of conformal observers, characterized by $\varphi<0$, is related to the trapping process which leads to the formation of black holes. Another relevance of the factor was highlighted in \cite{ojak}, where it plays a crucial role in a massless scalar field collapse. In this particular scenerio, in the vacinity of the singularity the spacetime is self-similar, i.e. $\varphi$ is constant (more precisely, its variation is negligible) in that region. 

As is seen, the factor $\varphi$ measures the deviation from harmonicity, of the potential function $\Psi$. It indeed holds true for a CKV generating a diffeomorphism, the factor $\varphi$ obeys the Klein-Gordon equation for a mass scalar field  (See Appendix B for the explicit derivation of the statement), i.e. $\Box\varphi=0$. This therefore suggest a generic relationship between the conformal factor and massless scalar field collapse in the presence of conformal symmetry, setting up the possibility for more general considerations of similar collapse studies in LRS spacetimes. (We note that the vanishing condition $\Box\varphi=0$ imposes $\Box^2\Psi=0$ on the potential function, where we have used $\Box^2$ as shorhand for twice applying the operator.)


\section*{Acknowledgements}


We are very grateful to the anonymous referee for pointing out relevant references and providing many helpful suggestions, which have consequently improved various aspects of the manuscript. The authors acknowledge that this research was supported by the Basic Science Research Program through the National Research Foundation of Korea (NRF) funded by the Ministry of education (grant numbers) (NRF-2022R1I1A1A01053784) and (NRF-2021R1A2C1005748). We would like to thank Rituparno Goswami of the University of KwaZulu-Natal for some helpful discussions.


\section*{Appendix A}\label{app}


One can write down an equivalent form of the Einstein field equations for a LRS spacetime as a set of evolution and propagation equations of the covariant variables specifying the spacetime, along the observers' and spatial congruences, along with a constraint equation specifying the magnetic Weyl scalar. These equations are obtained using appropriate contractions of the Ricci identities for the preferred unit directions. Extensive details can be found in the references \cite{cc1,cc2}, covering the general decomposition procedure.

\begin{itemize}

\item \textit{Evolution}

\begin{subequations}
\begin{align}
\frac{2}{3}\dot{\Theta}-\dot{\Sigma}&=\mathcal{A}\phi-\frac{1}{2}\left(\frac{2}{3}\Theta - \Sigma\right)^2+2\Omega^2-\frac{1}{3}\left(\rho+3p-2\Lambda\right)+\mathcal{E}-\frac{1}{2}\Pi,\label{evo1}\\
\dot{\phi}&=\left(\frac{2}{3}\Theta-\Sigma\right)\left(\mathcal{A}-\frac{1}{2}\phi\right)+2\xi\Omega+Q,\label{evo100}\\
\dot{\Omega}&=\mathcal{A}\xi-\left(\frac{2}{3}\Theta-\Sigma\right)\Omega,\label{aevo1}\\
\dot{\xi}&=-\frac{1}{2}\left(\frac{2}{3}\Theta-\Sigma\right)\xi+\frac{1}{2}\mathcal{H}+\left(\mathcal{A}-\frac{1}{2}\phi\right)\Omega,\label{aevo2}\\
\dot{\mathcal{E}}-\frac{1}{3}\dot{\rho}+\frac{1}{2}\dot{\Pi}&=-\left(\frac{2}{3}\Theta-\Sigma\right)\left(\frac{3}{2}\mathcal{E}+\frac{1}{4}\Pi\right)+\frac{1}{2}\phi Q+\frac{1}{2}\left(\rho+p\right)\left(\frac{2}{3}\Theta-\Sigma\right)+3\xi\mathcal{H},\label{evo101}\\
\dot{\mathcal{H}}&=-3\xi\mathcal{E}-\frac{3}{2}\left(\frac{2}{3}\Theta-\Sigma\right)\mathcal{H}+\Omega Q+\frac{3}{2}\xi\Pi.\label{aevo3}
\end{align}
\end{subequations}

\item \textit{Propagation}

\begin{subequations}
\begin{align}
\frac{2}{3}\hat{\Theta}-\hat{\Sigma}&=\frac{3}{2}\phi \Sigma+2\xi\Omega +Q,\label{evo2}\\
\hat{\phi}&=\left(\frac{1}{3}\Theta+\Sigma\right) \left(\frac{2}{3}\Theta-\Sigma\right)-\frac{1}{2}\phi^2+2\xi^2-\frac{2}{3}\left(\rho+\Lambda\right)-\mathcal{E}-\frac{1}{2}\Pi,\label{evo200}\\
\hat{\Omega}&=\left(\mathcal{A}-\phi\right)\Omega,\label{aevo4}\\
\hat{\xi}&=-\phi\xi-\left(\frac{1}{3}\Theta+\Sigma\right)\Omega,\label{aevo5}\\
\hat{\mathcal{E}}-\frac{1}{3}\hat{\rho}+\frac{1}{2}\hat{\Pi}&=-\frac{3}{2}\phi\left(\mathcal{E}+\frac{1}{2}\Pi\right)+3\Omega\mathcal{H}-\frac{1}{2}\left(\frac{2}{3}\Theta-\Sigma\right)Q,\label{evo201}\\
\hat{\mathcal{H}}&=-\left(3\mathcal{E}+\rho+p-\frac{1}{2}\Pi\right)\Omega-\frac{3}{2}\phi\mathcal{H}-Q\xi.\label{aevo6}
\end{align}
\end{subequations}

\item \textit{Propagation/Evolution}

\begin{subequations}
\begin{align}
\hat{\mathcal{A}}-\dot{\Theta}&=-\left(\mathcal{A}+\phi\right)\mathcal{A}+\frac{1}{3}\Theta^2+\frac{3}{2}\Sigma^2+\frac{1}{2}\left(\rho+3p-2\Lambda\right),\label{evo3}\\
\hat{Q}+\dot{\rho}&=-\Theta\left(\rho+p\right)-\left(\phi+2\mathcal{A}\right)Q-\frac{3}{2}\Sigma\Pi,\label{evo300}\\
\hat{p}+\hat{\Pi}+\dot{Q}&=-\left(\frac{3}{2}\phi+\mathcal{A}\right)\Pi-\left(\frac{4}{3}\Theta+\Sigma\right)Q-\left(\rho+p\right)\mathcal{A},\label{evo301}
\end{align}
\end{subequations}
with \(\Lambda\) being the cosmological constant.
\end{itemize}

We also have the following constraint:

\begin{eqnarray}\label{aevo10}
\mathcal{H}=3\Sigma\xi-\left(2\mathcal{A}-\phi\right)\Omega.
\end{eqnarray}

\section*{Appendix B}\label{app}

We establish that the conformal factor $\varphi$ associated to a CKV, where the CKV generates a diffeomorphism, obeys the wave equation $\Box\varphi=0$.

For any generator $V^{\mu}$ of diffeomorphism,

\begin{eqnarray}
J_{\mu}=\nabla^{\nu}\nabla_{\mu}V_{\nu}-\nabla^{\nu}\nabla_{\nu}V_{\mu},
\end{eqnarray}
defines a conserved quantity: $\nabla^{\mu}J_{\mu}=0$, as can be easily verified. This is the Komar current \cite{komar1}. 

Now, $J_{\mu}$ is just the right hand side of the contracted Ricci identities:

\begin{eqnarray}
J_{\mu}=R_{\mu\nu}V^{\nu},
\end{eqnarray}
Thus, if the field $V^{\mu}$ is a CKV, then taking the divergence of the above, making use of the contracted Bianchi identities as well as noting the symmetricity of the Ricci tensor, we have that

\begin{eqnarray}\label{box}
V^{\mu}\nabla_{\mu}R+2\varphi R=0.
\end{eqnarray}

Now, from the Ricci identities for $V^{\mu}$ and using the definition of $\varphi$, it can be established that 

\begin{eqnarray}
\Box\varphi=-\frac{1}{6}V^{\mu}\nabla_{\mu}R-\frac{1}{3}\varphi R.
\end{eqnarray}
Then, it follows that simply comparing the above to \eqref{box} gives the desired $\Box\varphi$=0.


\end{document}